\documentclass[A4paper,12pt]{article} 

\usepackage{amsmath,amsthm,amsfonts,amssymb}
\usepackage{amsbsy,latexsym,amsopn,amstext,amsxtra,euscript,amscd}
\usepackage{array}
\usepackage{ifthen}
\usepackage{url}
\usepackage{graphicx}
\usepackage{booktabs}
\usepackage{multirow}
\usepackage{float}

\def\8u{\infty}

\newcommand{\dm}{\begin{displaymath}}
\newcommand{\md}{\end{displaymath}}
\newcommand{\ct}{\begin{center}}
\newcommand{\tc}{\end{center}}

\numberwithin{equation}{section} 

\newtheorem{theorem}{Theorem}
\newtheorem{definition}{Definition}
\newtheorem{lemma}{Lemma}
\newtheorem{remark}{Remark}
\newtheorem{corollary}{Corollary}

\begin{document}

\title{{\bf  Sparse System Identification in Pairs of FIR and TM Bases}}
\author{
     \quad\ {\bf Dan Xiong$^1$, Li Chai$^2$, Jingxin Zhang$^3$ }\\
{\small 1. School of Information Science and Engineering} \\
{\small 2. Engineering Research Center of Metallurgical Automation and Measurement Technology}\\
{\small Wuhan University of Science and Technology, Hubei, Wuhan, 430081, China}\\
{\small 3. School of Software and Electrical Engineering}\\
{\small Swinburne University of Technology, Melbourne, VIC3122, Australia}}
\date{} 
\maketitle
\pagestyle{empty}  
\thispagestyle{empty} 

{\hspace{0mm}\bf {Abstract}}:\quad
This paper considers the reconstruction of a sparse coefficient vector $\theta$ for a rational transfer function,
under a pair of FIR and Takenaka-Malmquist (TM) bases and from a limited number of linear frequency-domain measurements.
We propose to concatenate a limited number of FIR and TM basis functions in the representation of the transfer function,
and prove the uniqueness of the sparse representation defined in the infinite dimensional function space
with pairs of FIR and TM bases.
The sufficient condition is given for replacing the $\ell_0$ optimal solution by the $\ell_1$ optimal solution using FIR and TM bases with random samples on the upper unit circle, as the foundation of reconstruction.
The simulations verify that $\ell_1$ minimization can reconstruct the coefficient vector $\theta$ with high probability.
It is shown that the concatenated FIR and TM bases give a much sparser representation,
with much lower reconstruction order than using only FIR basis functions and less dependency on the knowledge of the true system poles than using only TM basis functions.


\indent{\hspace{0mm}\bf Key words}:\quad sparse system identification; FIR basis; TM basis; $\ell_1$ optimization.


\section{Introduction}

{\hspace{5mm}
System identification has a long history in control theory. 
System identification using finite impulse response (FIR) model has been studied for many years.
FIR modeling corresponds to estimating the expansion coefficients of a partial expansion
in the standard orthonormal function basis $\{z^{-k}\}$.
The main advantage of FIR model is that the parameters
(the impulse response coefficients) appear linearly in the model,
leading to a simple estimation problem.
Many excellent works on the parameter estimation have been done \cite{Lev:60}, \cite{RaCrAl:78}
and some methods and algorithms have been well developed,
for example, the least mean square \cite{WiMcLa:76}, \cite{HuKa:94} and subspace identification \cite{Katayama:05} methods.
Although all these methods can achieve effectively system identification,
the main disadvantage of the FIR model is that in general one needs to estimate a large amount of
expansion coefficients if the pole of the transfer function is close to the unit circle
and hence the impulse response decays slowly, which will lead to a very high order reconstruction.
To overcome this problem, system identification using rational orthonormal basis functions with structures was introduced.

Between the 1920s and 1990s,
most of the research on using rational orthonormal basis functions was focused on
the construction using Laguerre functions with a single
(repeating for $k>1$) real pole $a \in (-1,1)$ \cite{Wah:13}, \cite{MaEkRa:98}, \cite{ChVe:99},
and Kautz functions having two complex conjugate poles \cite{Wah:94}, \cite{WaMa:96}.
Ninness et al. used an arbitrary sequence of poles, which led to the Takenaka-Malmquist (TM) basis functions \cite{NiGu:97}.
Since the work of \cite{NiGu:97}, the Generalized Orthogonal Basis function (GOBF) based construction was introduced into
the arena of systems.
Over the last twenty years, identification and control of linear stable dynamic systems using Orthonormal Rational Functions
(ORFs) have been widely used, see for instance
\cite{AkNi:98, Nin:96, WaPa:96, GuBu:03, NiGoWe:02, MiQi:12, ChMaZh:15, Oht:11, HiMe:12, TiSc:14}.

An LTI system can be well approximated with a small number of ORFs
if the poles in the TM basis are close to the true system poles \cite{Heu:05}.
If sufficient samples of the transfer function on the unit circle are acquired,
the coefficient estimation can be solved by least squares (LS) method.
However, it is generally difficult to identify the order and the poles of the transfer function in advance.
This limits the usefulness of ORF based system identification methods and
has inspired compressed sensing based FIR system identification in recent years.

Compressed sensing (CS)
\cite{Don:06}, \cite{CaRoTa:06a}, \cite{CaTa:06}, \cite{CaRoTa:06b}
is a new framework for simultaneous sampling and compression of signals.
It has drawn much attention since its advent several years ago,
and has been applied to the identification of sparse systems.
Sparse system identification using the least mean square (LMS) algorithm was discussed in
\cite{GuJiMe:09}, \cite{ChGuHe:09}, \cite{KaMiBa:11},
and the algorithm based on the projections onto weighted $\ell_1$ balls was proposed in \cite{KoSlTh:10}, \cite{SlKoTh:10}.
The essence of these methods is to find a sparse representation of the system
confined to a single basis and may not yield the sparsest solution.

It is well known from the CS literature that a signal may have a much sparser representation in an
overcomplete basis (redundant dictionary) consisting of concatenated orthogonal bases
\cite{ElBr:02}, \cite{DoEl:03}, \cite{GrNi:03}, \cite{MaCeWi:04}.
In the context of finite dimensional vector spaces,
\cite{ElBr:02} and \cite{DoEl:03} have presented and analyzed
the sparse representation of vector signals under a pair of orthonormal bases.
In the context of finite dimensional function spaces, \cite{Rau:10} has discussed the random sampling in
the bounded orthonormal systems with one orthonormal basis from the perspective of structured random matrix.

Inspired by these works,
this paper investigates the sparse system identification under a pair of FIR and TM bases
and from a limited number of linear measurements.
Here the identification is to reconstruct a sparse coefficient vector $\theta$ for a rational transfer function under such pairs.
In other word, the coefficient vector $\theta$ is the object of the identification.
The aim is to obtain a sparse representation with much fewer significant coefficients than using only FIR basis functions
and with weaker dependence on the true system poles than using only the ORFs,
and hence to overcome the drawbacks of these two types of bases.

Based on the analysis of \cite{XiChZh:17},
we show the uniqueness of sparse representation
for rational transfer functions in the infinite dimensional function space
with pairs of FIR and TM bases,
using the uniform bound of maximal absolute inner product of such pairs as an index.
We then derive a compressed sensing formulation for
finding the sparse representation of the rational transfer function
in the concatenated FIR and TM bases.
We further show that the replacement of $\ell_0$ optimization by $\ell_1$ optimization with randomly sampled frequency domain measurements and under a pair of FIR and TM bases is guaranteed with high probability.
Numerical experiments verify the effectiveness of the proposed identification framework.

The contributions of this paper are:

$\bullet$ Analysis on the uniqueness property of sparse representation of rational transfer function in the pairs of FIR and TM bases.

$\bullet$ A novel identification method for rational transfer function with the finite-order combination of FIR basis and TM basis.

$\bullet$ Sufficient conditions on the number of measurements needed to recover the sparse coefficient from the randomly sampled measurements by solving the $\ell_1$-minimization problem in the pairs of FIR and TM bases
and demonstration of the reconstruction performance of the proposed method.

The rest of this paper is organized as follows.
In Section 2,
the uniqueness of the sparse representation of rational transfer functions in the pairs of FIR and TM bases is given.
The sparse system identification using concatenated FIR and TM bases is given in Section 3.
Section 4 discusses computation issues of the proposed method.
Section 5 presents the simulation results, followed by conclusions in Section 6.

\section{Sparse representation of transfer functions in pairs of FIR and TM bases and Uniqueness Property}
{\hspace{5mm}
Let $H(z)$ be a proper, stable, real-rational transfer function with at least one nonzero pole.
Assume that $H(z)$ has a ``sparse'' representation under a pair of ORF bases,
$\{\phi_{k}(z), k=1, 2, \cdots, \}$ and $\{\psi_{l}(z), l=1, 2, \cdots\}$,
that is
$$
H(z)=\sum\limits_{k=1}^{\infty} \alpha_k \phi_k(z) + \sum\limits_{l=1}^{\infty} \beta_l \psi_{l}(z).
$$
In this section, we will show the uniqueness property of the sparse representation
with real valued sparse coefficients
$\alpha = [\alpha_1, \alpha_2, \cdots]^T$ and
$\beta = [\beta_1, \beta_2, \cdots]^T$.
Here the sparsity of $\alpha$ and $\beta$ is in the sense that
$\|\alpha\|_{0(\varepsilon)} \leq s_1$, $\|\beta\|_{0(\varepsilon)} \leq s_2$,
which means the rational transfer function $H(z)$ is $(\varepsilon, s_1+s_2)$-sparse in the pairs of orthonormal rational functions,
where $\|\cdot\|_{0(\varepsilon)}$ is the $\varepsilon$-0 norm defined as follows.

\begin{definition} \label{es}
For a fixed threshold $\varepsilon >0$ and an infinite sequence $\alpha=[\alpha_1, \alpha_2, \cdots]^T$
satisfying $\sum\limits_{k=1}^{\infty} |\alpha_k| < \infty$,
let
$$
N_{\varepsilon}(\alpha)=\min \{K: \sum_{k=K}^{\infty}|\alpha_k|\leq \varepsilon\}
$$
and define the
$\varepsilon$-support of $\alpha$ as
$$\Gamma_{\varepsilon}(\alpha)=\{k: |\alpha_k|\neq 0, 1\leq k < N_{\varepsilon}(\alpha)\},$$
and the cardinality of \ $\Gamma_{\varepsilon}(\alpha)$ as the $\varepsilon$-0 norm of $\alpha$,
denoted by $\| \alpha \|_{0(\varepsilon)}$.
\end{definition}

In this paper, we focus on the concatenation of FIR basis
$$\phi_k(z)=z^{-(k-1)},\quad k=1, 2, \cdots $$
and Takenaka-Malmquist basis (TM basis) \cite{Tak:25}, \cite{Mal:25}
\begin{equation} \label{tmb}
\psi_l(z):=\frac{\sqrt{1-|\xi_l|^2}}{z-\xi_l}\prod_{j=1}^{l-1}\frac{1-\bar{\xi}_j z}{z-\xi_j},\quad l=1, 2, \cdots,
\end{equation}
where the poles $\{\xi_{l}\} \subset \mathbb{D}=\{z|\,|z|< 1\}$ are given
and $\bar{\xi}_j$ is the complex conjugate of $\xi_j.$

The construction of TM basis holds for multiple poles and complex poles as well.
If any of the poles $\{\xi_l\}$ are chosen as complex, then the impulse responses of TM basis are complex-valued,
which is inappropriate. However, the construction of new basis functions which have the same complex poles but have real valued impulse responses can solve this problem, see \cite{NiGoWe:02} for details.
In addition, the necessary and sufficient condition for the completeness of TM basis functions is
$\sum\limits_{l=1}^{\infty} (1-|\xi_l|)=\infty$ \cite{NiGu:97}.

Both bases are orthonormal in terms of the inner product
\begin{equation} \label{innerproduct}
\langle \phi_k(z), \phi_{k'}(z) \rangle
=\frac{1}{2\pi i} \oint_\mathbb{T} \phi_k(z) \overline{\phi_{k'}(z) }\frac{dz}{z}
=\frac{1}{2\pi} \int_0^{2\pi} \phi_k(e^{i\omega})\overline{\phi_{k'}(e^{i\omega})}d \omega,
\end{equation}
where $\mathbb{T}=\{z|\,|z|=1\}.$

We have presented in \cite{XiChZh:17} the uniqueness of sparse representation of transfer function
in pairs of general ORF bases with the representation coefficients satisfying
$$
{(\sqrt{\| \alpha \|_{0(\varepsilon)} }
+\varepsilon)}^2+{(\sqrt{\| \beta \|_{0(\varepsilon)} }+ \varepsilon)}^2
< \frac{1}{\mu},
$$
where
$\mu=\sup_{k,l} |\langle \phi_k(z), \psi_l(z)\rangle|$
is the mutual coherence of such two ORF bases.
The concept of mutual coherence for matrices was introduced by David Donoho and Michael Elad \cite{DoEl:03}.
The mutual coherence has been used extensively in the field of sparse representations of signals
since it is a key measure of the bound for the unique representation of a sparse signal
and the ability of suboptimal algorithms such as matching pursuit and basis pursuit to correctly identify the sparse signal.
Following the terminology of compressed sensing,
we denote $\mu$ the mutual coherence of two ORF bases $\{\phi_k(z)\}$ and $\{\psi_l(z)\}$.

Notice that FIR and TM bases are two special cases of ORF bases.
Express the TM basis in impulse response
\begin{equation} \label{Eq3.1}
\psi_l(z):=\sum_{d=0}^{\infty}a_{dl}z^{-d}, l=1,2,\cdots.
\end{equation}
Then the inner product of $\phi_k(z)$ and $\psi_l(z)$ is given by
$$
\langle \phi_k(z), \psi_l(z)\rangle
=\langle z^{-(k-1)}, \sum_{d=0}^{\infty}a_{dl}z^{-d}\rangle
= \sum_{d=0}^{\infty} \langle z^{-(k-1)}, z^{-d}\rangle a_{dl}
= a_{(k-1),l}.
$$
The last equation follows from the orthonormality of FIR bases $\{z^{-(k-1)}\}_{k=1}^{\infty}$.
Hence the mutual coherence of FIR and TM bases is
$$
\mu=\sup_{k,l} |\langle \phi_k(z), \psi_l(z)\rangle|=\sup_{k,l} |a_{(k-1),l}|=\sup_{d,l} |a_{dl}|.
$$
Denote the uniform bound of the maximal absolute impulse response of TM basis as
\begin{equation} \label{Eq3.2}
\tilde{\mu}=\sup_{d,l} |a_{dl}|.
\end{equation}
Using a similar proof of Theorem 2 in \cite{XiChZh:17},
we can establish the uniqueness property of the sparse representation of transfer function in pair of FIR and TM bases.

\begin{theorem} \label{FTU}
For a transfer function $H(z)$ with a representation in the concatenated FIR and TM bases
$$
H(z)=\sum_{k=1}^\infty \alpha_k z^{-(k-1)}+\sum_{l=1}^\infty \beta_l \psi_l(z),
$$
where $\psi_l(z)$ is given in (\ref{tmb}).
For a fixed thresholds $\varepsilon > 0$,
if the representation is sparse in the sense of $\varepsilon$-0 norm and
$$
{(\sqrt{\| \alpha \|_{0(\varepsilon)} }
+\varepsilon)}^2+{(\sqrt{\| \beta \|_{0(\varepsilon)} }+ \varepsilon)}^2
< \frac{1}{\tilde{\mu}}
$$
with $\tilde{\mu}$ as defined in (\ref{Eq3.2}),
then this sparse representation is unique.
\end{theorem}

\begin{remark}
For the pair of general ORF bases, if the numbers of such two bases are given as $n_1$ and $n_2$, respectively,
then the complexity of the mutual coherence $\mu$ of such two bases is $O(n_1n_2)$.
However from (\ref{Eq3.2}), for the pair of FIR and TM bases,
the analytic formula of the mutual coherence $\tilde{\mu}$ is given,
which shows that $\tilde{\mu}$ only depends on the maximal absolute impulse response of TM basis functions,
thus the complexity of $\tilde{\mu}$ is the number of TM basis functions $O(n_2)$.
That is, the mutual coherence of FIR and TM bases is easier to compute than the general ORF bases.
\end{remark}

When $\varepsilon=0$, the $\varepsilon$-0 norm reduces to the standard definition of 0-norm.
Then we have the following Corollary.

\begin{corollary}
If a transfer function $H(z)$ has a sparse representation in the concatenated FIR and TM bases
$$
H(z)=\sum_{k=1}^\infty \alpha_k z^{-(k-1)}+\sum_{l=1}^\infty \beta_l \psi_l(z)
\ \mbox{and} \
\| \alpha \|_{0}+\| \beta \|_{0} < \frac{1}{\tilde{\mu}}
$$
with $\tilde{\mu}$ as defined in (\ref{Eq3.2}),
then this sparse representation is unique.
\end{corollary}

The impulse responses $\{a_{dl}\}$ in (\ref{Eq3.1}),
which determine the value of $\tilde{\mu}$,
can be obtained by the following theorem for $\psi_l(z)$ with distinct poles.

\begin{theorem}
The Laurent expansion of the TM basis function $\psi_{l}(z)$ with distinct poles $\xi_j (j=1, 2, \cdots, l)$ in the annulus
$\{z|\max\limits_{1 \le j \le l} |\xi_j| < |z|< 2\}$ is
$$\psi_{l}(z):=\sum_{d=0}^{\infty}a_{dl}z^{-d}, \quad l=1, 2, \cdots, $$
where
\[
a_{dl}
= \sqrt{1-|\xi_l|^2} \sum_{j'=1}^l \xi_{j'}^{d-1} \frac{\prod_{j=1}^{l-1}(1-\bar{\xi}_j \xi_{j'})}{\prod_{j=1, j\neq j'}^{l}(\xi_{j'}-\xi_j)},
\]
for $d=1, 2, \cdots,$
and
$a_{0l}=0$ for all $l$.
\end{theorem}
\begin{proof}
In the annulus $\{z|\max\limits_{1 \le j \le l} |\xi_j| < |z|< 2\}$,
from the Laurent expansion, for $d=1, 2 , \cdots$

1) $l=1$
\begin{eqnarray*}
a_{d1}
&=&\frac{1}{2\pi i} \oint_C \frac{\psi_1(z)}{z^{-d+1}}dz
= Res\left[\frac{\psi_1(z)}{z^{-d+1}},\xi_{1}\right]
= \lim_{z\rightarrow \xi_{1}}(z-\xi_{1}) z^{d-1} \frac{\sqrt{1-|\xi_1|^2}}{z-\xi_1}\\
&=& \sqrt{1-|\xi_1|^2} {\xi_1}^{d-1},
\end{eqnarray*}
where Res$[\cdot, \cdot]$ denotes the Residue of a complex function.

2) $l\geq 2$
\begin{eqnarray*}
a_{dl}&=&\frac{1}{2\pi i} \oint_C \frac{\psi_l(z)}{z^{-d+1}}dz
= \sum_{j'=1}^l Res \left[\frac{\psi_l(z)}{z^{-d+1}},\xi_{j'} \right]\\
&=&\sum_{j'=1}^l Res \left[z^{d-1} \frac{\sqrt{1-|\xi_l|^2}}{z-\xi_l}\prod_{j=1}^{l-1}\frac{1-\bar{\xi}_j z}{z-\xi_j},\xi_{j'} \right]\\
&=& \sum_{j'=1}^l \lim_{z\rightarrow \xi_{j'}}(z-\xi_{j'}) z^{d-1} \frac{\sqrt{1-|\xi_l|^2}}{1-\bar{\xi}_l z}\prod_{j=1}^{l}\frac{1-\bar{\xi}_j z}{z-\xi_j}\\
&=& \sum_{j'=1}^l \xi_{j'}^{d-1} \frac{\sqrt{1-|\xi_l|^2}}{1-\bar{\xi}_l \xi_{j'}}\frac{\prod_{j=1}^{l}(1-\bar{\xi}_j \xi_{j'})}{\prod_{j=1, j\neq j'}^{l}(\xi_{j'}-\xi_j)}\\
&=& \sqrt{1-|\xi_l|^2} \sum_{j'=1}^l \xi_{j'}^{d-1} \frac{\prod_{j=1}^{l-1}(1-\bar{\xi}_j \xi_{j'})}{\prod_{j=1, j\neq j'}^{l}(\xi_{j'}-\xi_j)}.
\end{eqnarray*}
Denote $\frac{\prod_{j=1}^{l-1}(1-\bar{\xi}_j \xi_{j'})}{\prod_{j=1, j\neq j'}^{l}(\xi_{j'}-\xi_j)}=1$ for $l=1$.
Then the above derivations can be unified as
$$
a_{dl}
= \sqrt{1-|\xi_l|^2} \sum_{j'=1}^l \xi_{j'}^{d-1} \frac{\prod_{j=1}^{l-1}(1-\bar{\xi}_j \xi_{j'})}{\prod_{j=1, j\neq j'}^{l}(\xi_{j'}-\xi_j)},
$$
for $l, d=1, 2 , \cdots.$
Further, it is obvious that $a_{0l}=0$ for all $l$.
\end{proof}

\section{Sparse System Identification using concatenated FIR and TM bases}
{\hspace{5mm}
Theorem \ref{FTU} shows that if the representation coefficient vector
$
\theta=[\alpha^T \quad \beta^T]^T
$
is sparse enough,
the representation is unique.
In this section,
we will propose a reconstruction algorithm based on compressed sensing
to reconstruct $H(z)$ with a small fraction of the measurements of $H(z)$ on the unit circle.
Precisely, define
$$
T_N:=\{z_{r}=e^{2\pi i (r-1)/N},\ r=1,2,\cdots,N\}.
$$
We focus on the underdetermined case
with only a few of the components of $\{H(z_r),  r=1, 2, \cdots, N\}$ sampled or observed.
That is, only a small fraction of $T_N$ is known.
Given a subset $\Omega \subset \{1,2,\cdots,N\}$ of size $|\Omega|=m$ (far less than the number of the basis functions),
the goal is to reconstruct the representation coefficients and hence the transfer function $H(z)$
from the much shorter $m$-dimensional measurements $\{H(z_r), r \in \Omega\}.$

Now we will restate this problem in a matrix form.
Combining with Definition 1,
$H(z)$ can be rewritten as
$$
H(z)
=\sum\limits_{k=1}^{n_1} \alpha_k z^{-(k-1)} + \sum\limits_{l=1}^{n_2} \beta_l \psi_{l}(z) + \Delta_1 +\Delta_2,
$$
where $n_1=N_{\varepsilon}(\alpha)-1$ and $n_2=N_{\varepsilon}(\beta)-1$,
$\Delta_1=\sum\limits_{k=n_1+1}^{\infty} \alpha_k z^{-(k-1)}$
and
$\Delta_2=\sum\limits_{k=n_2+1}^{\infty} \beta_l \psi_l(z)$.

By simple calculation, we have
$$
\|\Delta_1\|^2
= \langle \sum\limits_{k=n_1+1}^{\infty}  \alpha_k z^{-(k-1)}, \sum\limits_{k=n_1+1}^{\infty} \alpha_k z^{-(k-1)}\rangle
= \sum\limits_{k=n_1+1}^{\infty}  |\alpha_k|^2
\leq \varepsilon^2.
$$
Similarly, we have
$
\|\Delta_2\|^2 \leq \varepsilon^2.
$

Denote $\Delta=\Delta_1+\Delta_2$, then we have
$$
\|\Delta\| \leq \sqrt{2(\|\Delta_1\|^2 + \|\Delta_2\|^2)} \leq 2 \varepsilon.
$$
Now the transfer function $H(z)$ can be simplified as
\begin{equation} \label{simplever}
H(z)
=\sum\limits_{k=1}^{n_1} \alpha_k z^{-(k-1)} + \sum\limits_{l=1}^{n_2} \beta_l \psi_{l}(z) + \Delta,
\end{equation}
with $\|\Delta\| \leq 2 \varepsilon$.
With a little bit abuse of the notation,
the unknown coefficients to be determined here are denoted as
$\alpha=[\alpha_1, \alpha_2, \cdots, \alpha_{n1}]^T$ and
$\beta=[\beta_1, \beta_2, \cdots, \beta_{n2}]^T$.

Due to the arbitrariness of $\varepsilon$, the norm of the term $\Delta$ can be arbitrarily small,
and the term $\Delta=0$ when $\varepsilon$ is exactly zero.
In the sequel, we first discuss the equation (\ref{simplever}) with the term $\Delta$ omitted.

Define
$[\Phi \quad \Psi]$ to be a composite sample matrix with its $r$-th ($r=1, 2, \cdots, N$) row satisfying
\begin{equation} \label{total sample matrix}
[\Phi \quad \Psi]_r:= [1, z_r^{-1},\cdots, z_r^{-n_1+1}, \psi_1(z_r), \cdots, \psi_{n_2}(z_r)]
\end{equation}
and
\begin{equation*} 
H:=[H(z_1), H(z_2), \cdots, H(z_N)]^T.
\end{equation*}
\noindent
Then
\begin{equation*} 
H=
[\Phi \quad \Psi]
\begin{bmatrix}
\alpha\\
\beta
\end{bmatrix},
\end{equation*}
where
$$
\Phi=
\begin{bmatrix}
1 & z_1^{-1} & \cdots & z_1^{-n_1+1} \\
1 & z_2^{-1} & \cdots & z_2^{-n_1+1}\\
\cdots & \cdots & \cdots & \cdots\\
1 & z_N^{-1} & \cdots & z_N^{-n_1+1}
\end{bmatrix}.
$$

\noindent
We randomly select the subset $\Omega$ of size $m (<<n_1+n_2)$ drawn from
the uniform distribution over the index set $\{1,2,\cdots,N\}$,
and denote the measurement by
$$
H_{\Omega}=[\Phi \quad \Psi]_{\Omega}
\begin{bmatrix}
\alpha\\
\beta
\end{bmatrix},
$$
where $H_{\Omega}$ is the $m \times 1$ vector consisting of $\{H(z_r),  r \in \Omega\}$,
and $[\Phi \quad \Psi]_{\Omega}$ is the $m \times (n_1+n_2)$ matrix
with the $r$-th row $[\Phi \quad \Psi]_r, r \in \Omega$.

As $[\Phi \quad \Psi]_{\Omega}$ is the concatenation of two bases,
the representation is not unique in general.
However, as shown in Theorem 1,
if the representation is sufficiently sparse,
the uniqueness of the representation is guaranteed .
The goal is to find the sparsest representation from the $\ell_0$ minimization
$$
(P_0): \
\min_{\alpha, \beta} \left\|
\begin{bmatrix}
\alpha\\
\beta
\end{bmatrix}
\right \|_{0}
 \ \mbox{subject to} \  H_{\Omega} = [\Phi \quad \Psi]_{\Omega}
\begin{bmatrix}
\alpha\\
\beta
\end{bmatrix},
$$
which is an infeasible search problem \cite{Nat:06}.
An alternative approach is to solve the  $\ell_1$ minimization problem (Basis Pursuit)
\cite{Don:06}, \cite{CaRoTa:06a}, \cite{ChDoSa:01}
\begin{equation} \label{Eq4.3}
(P_1): \
\min_{\alpha, \beta} \left \|
\begin{bmatrix}
\alpha\\
\beta
\end{bmatrix}
\right \|_{1}
\  \mbox{subject to} \  H_{\Omega} = [\Phi \quad \Psi]_{\Omega}
\begin{bmatrix}
\alpha\\
\beta
\end{bmatrix},
\end{equation}
which can be solved by linear programming or second order cone program \cite{Rau:10}, \cite{BoVa:04}.

Further, taking the data with small perturbations into consideration,
the measurement is given by
$$
H_{\Omega}=[\Phi \quad \Psi]_{\Omega}
\begin{bmatrix}
\alpha\\
\beta
\end{bmatrix}  + \eta.
$$
Here the small perturbations $\eta$ can be either the transfer functions that are not exactly sparse but
nearly sparse (compressible), or the noise in the sampling process,
and is bounded by $\|\eta\|_2 \leq \epsilon$.
The corresponding approach is called Basis Pursuit Denoising (BPDN)
\begin{equation} \label{Eq4.5}
\min_{\alpha, \beta} \left \|
\begin{bmatrix}
\alpha\\
\beta
\end{bmatrix} \right \|_{1}
\ \mbox{subject to} \ \left \|H_{\Omega} - [\Phi \quad \Psi]_{\Omega}
\begin{bmatrix}
\alpha\\
\beta
\end{bmatrix} \right \|_2   \leq  \epsilon.
\end{equation}
For the sparse representation using only one basis,
compressed sensing theory has presented
the equivalence of $\ell_0$ optimization and $\ell_1$ minimization when the representation is sufficiently sparse \cite{Nat:06}, \cite{ChDoSa:01},
and has provided the sufficient conditions on the number of measurements needed to
recover the sparse coefficient from the randomly sampled measurements
by solving the $\ell_1$-minimization problem \cite{Rau:10}, \cite{Nat:06}.

For the setting of (\ref{Eq4.3}) concerning two bases,
\cite{XiChZh:17} has presented the lower bound of the number of measurements
which guarantees the replacement of $\ell_0$ optimization $(P_0)$ by $\ell_1$ optimization $(P_1)$
under a pair of general ORF bases
for a fixed (but arbitrary) support.
As FIR and TM bases are special cases of ORF bases,
the replacement holds as well.

To present the sufficient condition on the number of measurements required for the sparse
reconstruction by $\ell_1$ optimization in pairs of FIR and TM bases with random samples,
we first present the orthonormality property of $[\Phi \quad \Psi]$,
which directly determines the mutual coherence $\mu(\Phi, \Psi)$ of matrices $\Phi$ and $\Psi$.
$\mu(\Phi, \Psi)$, as a key index in the reconstruction, is defined as
\begin{equation} \label{muco}
\mu ( \Phi, \Psi )=\max_{k, l}\frac{|\langle \Phi_k, \Psi_l \rangle|}{\|\Phi_k\| \, \|\Psi_l\|},
\end{equation}
where $\Phi_k$, $\Psi_l$ denote the $k$-th, $l$-th column of $\Phi$ and $\Psi$, respectively.

\begin{theorem} \label{ThOrho}
When $N$ is sufficiently large, the composite sampling matrix $[\Phi \quad \Psi]$ satisfies:

(i) $\Phi^{*} \Phi \approx N I_{n_1}$,
where $^*$ is the conjugate transpose, $I_{n_1}$ is the identity matrix of dimension $n_1$.

(ii) $\Psi^{*} \Psi \approx N I_{n_2}$.

(iii) $
\Phi^{*} \Psi
= (\sum_{r=1}^N \psi_k(z_r) \overline{z_r^{-l}})
\approx
N(a_{k-1, l}),
\quad (k=1, \cdots, n_1, l=1, \cdots, n_2),
$
where $\{a_{k-1,l}\}$ are the impulse responses defined in equation (\ref{Eq3.1}).

(iv) $\mu(\Phi, \Psi)\approx \tilde{\mu}$.
\end{theorem}

\begin{proof}
The integral definition of inner product in (\ref{innerproduct}) shows that
when $N$ is sufficiently large,
$$
\frac{1}{2\pi} \sum_{r=1}^N z_r^{-(k-1)} \overline{z_r^{-(l-1)}}\frac{2\pi}{N}
\rightarrow
\langle z^{-(k-1)},z^{-(l-1)}\rangle
=\delta_{kl},
$$
where the kronecker symbol $\delta_{kl}$ equals 1 if $k=l$ and 0 if $k \neq l$.

Then the $(k, l)$ element of $\Phi^{*} \Phi$ is
$$
\sum_{r=1}^N \overline{z_r^{-(k-1)}}{z_r^{-(l-1)}}
=\overline{\sum_{r=1}^N z_r^{-(k-1)}\overline{{z_r^{-(l-1)}}}}
\rightarrow
N \delta_{kl},
$$
which implies (i).

The $(k, l)$ element of $\Psi^{*} \Psi$ is
\begin{eqnarray*}
\sum_{r=1}^N \overline{\psi_k(z_r)} \psi_l(z_r)
&=&\sum_{r=1}^N \overline{\sum_{d'=0}^{\infty} a_{d'k} z_r^{-d'}} \sum_{d=0}^{\infty} a_{dl} z_r^{-d}
=\sum_{d'=0}^{\infty} \sum_{d=0}^{\infty} a_{d'k} a_{dl} \sum_{r=1}^N \overline{z_r^{-d'}} z_r^{-d}\\
& \rightarrow &
N\sum_{d=0}^{\infty} a_{dk} a_{dl}
=N \delta_{kl},
\end{eqnarray*}
the last equation is based on the orthonormality of $\{\psi_l(z)\}$,
\begin{eqnarray*}
\langle \psi_k(z), \psi_l(z)\rangle
&=& \langle \sum_{d'=0}^{\infty}a_{d'k}z^{-d'}, \sum_{d=0}^{\infty}a_{dl}z^{-d} \rangle\\
&=& \sum_{d'=d}^{\infty} \langle a_{d'k}z^{-d'},a_{d'l}z^{-d} \rangle + \sum_{d' \neq d} \langle a_{d'k}z^{-d'},a_{dl}z^{-d} \rangle\\
&=& \sum_{d=0}^{\infty} a_{dk} a_{dl} \langle z^{-d}, z^{-d} \rangle + \sum_{d'\neq d} a_{d'k} a_{dl} \langle z^{-d'},z^{-d} \rangle\\
&=& \sum_{d=0}^{\infty} a_{dk} a_{dl}=\delta_{kl},
\end{eqnarray*}
which implies (ii).

As for (iii), the $(k, l)$ element of $\Phi^{*} \Psi$ is
$$
\sum_{r=1}^N \overline{z_r^{-(k-1)}} \psi_l(z_r)
=\sum_{r=1}^N \overline{z_r^{-(k-1)}}\sum_{d=0}^{\infty}a_{dl} z_r^{-d}
=\sum_{d=0}^{\infty}a_{dl} \sum_{r=1}^N \overline{z_r^{-(k-1)}} z_r^{-d}
\rightarrow
N a_{(k-1),l}.
$$
From the results (i)-(iii),
we show that when $N$ is sufficiently large,
$
\|\Phi_k\| \approx \sqrt{N} \ (k=1,2,\cdots, n_1)
$
and
$
\|\Psi_l\| \approx \sqrt{N} \ (l=1,2,\cdots, n_2)
$.
And
$
{\langle \Phi_k, \Psi_l \rangle},
$
corresponding to the $(k,l)$-element of $\Phi^*\Psi$, is approximately $N a_{(k-1),l}$.
Hence from the definition in (\ref{muco}),
the mutual coherence of $[\Phi \quad \Psi]$ constructed by FIR and TM bases
approximately equals $\max_{k, l}\{|a_{(k-1), l}|\}$,
i.e., $\tilde{\mu}$ in (\ref{Eq3.2}) when $N$ is sufficiently large.
\end{proof}
Denote $T_1=\{k: |\alpha_k| \neq 0\}$ and  $T_2=\{l: |\beta_l| \neq 0\}$
as the supports of coefficients $\alpha$ and $\beta$, respectively.
Let $\Phi_{T_1}$ be the $N \times |T_1|$ matrix corresponding to the columns of $\Phi$ indexed by $T_1$,
and define $\Psi_{T_2}$ similarly.

\begin{theorem}\label{Th6}
Fix a subset $T=T_1 \bigcup T_2$ of the coefficient domain,
with $T_1$ and $T_2$ being the supports of the coefficients $\alpha$ and $\beta$, respectively.
Choose a subset $\Omega$ of the measurement domain of size $|\Omega|=m$, and a sign sequence $\tau$ on $T$ uniformly at random. Suppose $m$ satisfies
$$
m \geq C \frac{1+\max\limits_{1\leq l \leq n_2}|\xi_l|}{1-\max\limits_{1\leq l \leq n_2}|\xi_l|} {\max}^2 \{|T|, \log{(\frac{n_1+n_2}{\delta})}, C_{\tilde{\mu},T,\delta}\},
$$
where
\begin{eqnarray*}
C_{\tilde{\mu},T,\delta}
=\frac{4}{\left[(\frac{1}{2} + \|\Phi_{T_1}^*\Psi_{T_2} /N\|){[2 \log{(\frac{2(n_1+n_2)}{\delta})}]}^{-\frac{1}{2}}  -  \tilde{\mu} \sqrt{|T|}\right]^2}.
\end{eqnarray*}
Then with the probability exceeding $1-6\delta$ and sufficiently large $N$,
every coefficient vector
$\theta=\begin{bmatrix}
\alpha\\
\beta
\end{bmatrix}$
supported on $T$ with sign matching $\tau$ can be recovered from solving the $\ell_1$ optimization problem
\begin{equation*}
(P_1): \
\min_{\hat{\alpha}, \hat{\beta}} \left \|
\begin{bmatrix}
\hat{\alpha}\\
\hat{\beta}
\end{bmatrix}
\right \|_{1}
\  \mbox{subject to} \  H_{\Omega}=[\Phi \quad \Psi]_{\Omega}
\begin{bmatrix}
\hat{\alpha}\\
\hat{\beta}
\end{bmatrix}
\end{equation*}
for the coefficient vector
$
\begin{bmatrix}
\hat{\alpha}\\
\hat{\beta}
\end{bmatrix}
$,
where
$H_{\Omega}= [\Phi \quad \Psi]_{\Omega}
\begin{bmatrix}
\alpha\\
\beta
\end{bmatrix}$.
\end{theorem}

\begin{proof}
From \cite{XiChZh:17}, the $\ell_1$ optimization can reconstruct
the sparse coefficient vector for general ORF pairs using random samples with high probability
when number of measurement $m$ satisfies
$$
m \geq C \mu_M^2 {\max}^2 \{|T|, \log{(\frac{n_1+n_2}{\delta})}, C_{\mu(\Phi, \Psi),T,\delta}\},
$$
where $\mu_M=\max\{\mu_\Phi, \mu_\Psi\}$ with $\mu_\Phi=\max|\Phi_{ij}|$ and $\mu_\Psi=\max|\Psi_{ij}|$
being the largest magnitude among the entries in $\Phi$ and $\Psi$, respectively.
And
$$
C_{\mu ( \Phi, \Psi ),T,\delta}
=\frac{4}{\left[(\frac{1}{2} + \|\Phi_{T_1}^*\Psi_{T_2} /N\|){[2 \log{(\frac{2(n_1+n_2)}{\delta})}]}^{- \frac{1}{2}}   -   \mu ( \Phi, \Psi ) \sqrt{|T|}\right]^2}.
$$

For the pair of orthonormal bases consisting of FIR and TM basis functions,
we have $\mu_{\Phi}=1$ since all elements in $\Phi$ has modulus equal to 1,
while for the second basis $\Psi$,
the element in $\Psi$ is $\Psi_l(z_r)$
with modulus equal to $\frac{\sqrt{1-|\xi_l|^2}}{|z_r-\xi_l|}$.
By using triangle inequality, we have
$$
1-|\xi_l|=
||z_r|-|\xi_l||
\leq |z_r-\xi_l|
\leq |z_r|+|\xi_l|=1+|\xi_l|,
$$
which implies
$$
\sqrt {\frac{1-|\xi_l|}{1+|\xi_l|}} \leq |\Psi_l(z_r)| \leq \sqrt {\frac{1+|\xi_l|}{1-|\xi_l|}}.
$$
Then
$$
\sqrt {\frac{1-\min\limits_{1\leq l \leq n_2}|\xi_l|}{1+\min\limits_{1\leq l \leq n_2}|\xi_l|}} \leq \mu_{\Psi} \leq \sqrt {\frac{1+\max\limits_{1\leq l \leq n_2}|\xi_l|}{1-\max\limits_{1\leq l \leq n_2}|\xi_l|}}.
$$
Thus $\mu_M \leq \sqrt {\frac{1+\max\limits_{1\leq l \leq n_2}|\xi_l|}{1-\max\limits_{1\leq l \leq n_2}|\xi_l|}}$.

According to Theorem \ref{ThOrho}, $\mu(\Phi, \Psi)\approx \tilde{\mu}$.
The claim then follows.
\end{proof}
Theorem \ref{Th6} shows that for most sparse coefficient vectors $\theta$ supported on a fixed (but arbitrary) set $T$,
the coefficient vectors can be recovered with overwhelming probability if the sign of $\theta$ on $T$ and the observations
$H_{\Omega}=[\Phi \quad \Psi]_\Omega \theta$ are drawn at random.

\section{Computation Issues}
{\hspace{5mm}
Notice that $H_{\Omega}$ and $[\Phi \quad \Psi]_{\Omega}$ in (\ref{Eq4.3}) are all complex-valued,
and compressed sensing with complex-valued data has been discussed in \cite{YuKhMa:12} and \cite{TiBo:13}.
A method adopted is to convert the $\ell_1$-norm minimization of complex signals to the second-order
cone programming (SOCP) \cite{MaCeWi:04}.
In this paper, we rewrite the complex-valued minimization in the real-valued form
by separating the real and imaginary parts of (\ref{Eq4.3}) and (\ref{Eq4.5}), respectively.
Denoting
$$
H_{\Omega}=H_{\Omega}^R+i H_{\Omega}^I,
$$
$$
[\Phi \quad \Psi]_{\Omega}=[\Phi \quad \Psi]_{\Omega}^R+i [\Phi \quad \Psi]_{\Omega}^I,
$$
$$
\eta=\eta^R+i \eta^I,
$$
we get the following equivalent optimization problems.

\begin{lemma}
The optimization (\ref{Eq4.3}) is equivalent to
\begin{equation} \label{Eq5.1}
\min_{\alpha, \beta}
\left \|
\begin{bmatrix}
\alpha\\
\beta
\end{bmatrix} \right \|_{1}
\
\mbox{subject to}
\
\begin{bmatrix}
H_{\Omega}^R\\ H_{\Omega}^I
\end{bmatrix}=
\begin{bmatrix}
[\Phi \quad \Psi]_{\Omega}^R \\ [\Phi \quad \Psi]_{\Omega}^I
\end{bmatrix}
\begin{bmatrix}
\alpha\\
\beta
\end{bmatrix},
\end{equation}
and the optimization (\ref{Eq4.5}) is equivalent to
\begin{equation}\label{Eq5.2}
\min_{\alpha, \beta}
\left \|
\begin{bmatrix}
\alpha\\
\beta
\end{bmatrix}  \right \|_{1}
\
\mbox{subject to}
\
 \left \|
\begin{bmatrix}
H_{\Omega}^R\\ H_{\Omega}^I
\end{bmatrix}-
\begin{bmatrix}
[\Phi \quad \Psi]_{\Omega}^R \\ [\Phi \quad \Psi]_{\Omega}^I
\end{bmatrix}
\begin{bmatrix}
\alpha\\
\beta
\end{bmatrix}
\right \|_2
\leq
\epsilon.
\end{equation}
\end{lemma}

\begin{proof}
$H_{\Omega}=[\Phi \quad \Psi]_{\Omega}
\begin{bmatrix}
\alpha\\
\beta
\end{bmatrix}$
is equivalent to
$$
H_{\Omega}^R+i H_{\Omega}^I=([\Phi \quad \Psi]_{\Omega}^R+i [\Phi \quad \Psi]_{\Omega}^I)
\begin{bmatrix}
\alpha\\
\beta
\end{bmatrix},
$$
which is equivalent to
\begin{equation*}
\begin{cases}
H_{\Omega}^R=[\Phi \quad \Psi]_{\Omega}^R
\begin{bmatrix}
\alpha\\
\beta
\end{bmatrix} \\
H_{\Omega}^I=[\Phi \quad \Psi]_{\Omega}^I
\begin{bmatrix}
\alpha\\
\beta
\end{bmatrix}\\
\end{cases}.
\end{equation*}
This proves (\ref{Eq5.1}).
Similar proof is also applicable to (\ref{Eq5.2}).
\end{proof}

In (\ref{Eq5.1}) and (\ref{Eq5.2}), the dimension of the problem is doubled.
For brevity, we refer to (\ref{Eq5.1}) and (\ref{Eq5.2}) as {\bf two-ortho model} hereafter.

On the other hand, the separation of the real and imaginary part of sensing matrix does not
change the mutual coherence of matrix $[\Phi \quad \Psi]$.
In fact, the following holds.

\begin{lemma}\label{Thmu}
For pairs of complex matrices $\Phi$ and $\Psi$,
$
\mu(\Phi, \Psi)
=\mu
\left(
\begin{bmatrix}
\Phi^R \\
\Phi^I
\end{bmatrix},
\begin{bmatrix}
\Psi^R \\
\Psi^I
\end{bmatrix}
\right),
$
where
$
\mu(\Phi, \Psi)
$
as defined in (\ref{muco}).
\end{lemma}

\begin{proof}
Based on the fact that $\Phi_k=\Phi_k^R+i \Phi_k^I$ and $\Psi_l=\Psi_l^R+i \Psi_l^I$, we have
$
\|\Phi_k\|
=\left\|
\begin{bmatrix}
\Phi^R_k \\
\Phi^I_k
\end{bmatrix}
\right\|
$
and
$
\|\Psi_l\|
=\left\|
\begin{bmatrix}
\Psi^R_l \\
\Psi^I_l
\end{bmatrix}
\right\|.
$

The inner products
$\langle \Phi_k, \Psi_l \rangle$ can be calculated by
the off-diagonal elements of the Gram matrix $[\Phi \quad \Psi]^*[\Phi \quad \Psi]$.
And the Gram matrix of
$
\begin{bmatrix}
\Phi^R  \quad \Psi^R \\
\Phi^I  \quad \Psi^I
\end{bmatrix}
$
is
$$
\begin{bmatrix}
\Phi^R  \quad \Psi^R \\
\Phi^I  \quad \Psi^I
\end{bmatrix}^T \!\!\!
\begin{bmatrix}
\Phi^R  \quad \Psi^R \\
\Phi^I  \quad \Psi^I
\end{bmatrix}
\!\!=\!\!
[\Phi^R \quad \Psi^R]^T [\Phi^R \quad \Psi^R]+[\Phi^I \quad \Psi^I]^T [\Phi^I \quad \Psi^I]
\!\!=\!\!
[\Phi \quad \Psi]^* [\Phi \quad \Psi],
$$
the last equality is based on $[\Phi \quad \Psi]=[\Phi \quad \Psi]^R+i [\Phi \quad \Psi]^I$.

It follows that
$
\mu(\Phi, \Psi)
=\mu
\left(
\begin{bmatrix}
\Phi^R \\
\Phi^I
\end{bmatrix},
\begin{bmatrix}
\Psi^R \\
\Psi^I
\end{bmatrix}
\right).
$
\end{proof}
Note that sampling on the upper unit circle instead of the entire unit circle is sufficient to reconstruct the coefficients due to the separation of the real and imaginary part of the model.
In fact, for the structure of matrix $[\Phi \quad \Psi]$ satisfying (\ref{total sample matrix}),
we have
$$\overline{z_r^{-(k-1)}}=(\bar{z}_r)^{-(k-1)}$$
and
$$\overline{\psi_l(z_r)}=\psi_l(\bar{z}_r)$$ for the real pole in the TM basis.
This means that
the $(N-r+2)$-th row of $\Phi$ is the conjugate of the $r$-th row
($r=2, 3, \cdots \frac{N}{2}$ when $N$ is even or $r=2, 3, \cdots \frac{N+1}{2}$ when $N$ is odd).
Hence, if we sample the conjugate rows to formulate (\ref{Eq5.1}),
the optimization will fail for the redundance in the equations.
Similarly, if the $z_r$ on the real line is sampled,
(\ref{Eq5.1}) will include a equation with all zero coefficients,
which will lead to the failure of (\ref{Eq5.1}) as well.
The strategy in this paper is to sample on the upper unit circle
 and exclude the endpoints, instead of the entire unit circle.

\section{Simulation}
{\hspace{5mm}
This section illustrates the use of the two-ortho model in the sparse system identification
with a simulation study.
Suppose an underlying transfer function is sparse under the pair of FIR and TM bases,
we implement the following experiments to demonstrate the reconstruction performance of the proposed algorithm.
The experiments include three aspects:
(a) the reconstruction of the sparse coefficients in two-ortho model;
(b) the reconstruction of underlying SISO transfer function with poles known;
(c) the reconstruction of underlying SISO transfer function with poles (of multiplicity greater than 1) partially known or unknown.

\subsection{Reconstruction of the sparse coefficients in the two-ortho basis representation}
{\hspace{5mm}
The following procedure is implemented to reconstruct the sparse coefficients of the transfer function
in the two-ortho basis representation.

Step 1:
Randomly generate an $(n_1+n_2)$-dimensional coefficient vector $\theta=[\alpha^T, \beta^T]^T$
consisting of $(s_1+s_2)$ spikes with amplitude $+1$,
and uniformly sample from the interval (0, 1) the poles $\{\xi_1,\xi_2,\cdots,\xi_{n_2}\}$ in the TM basis $\{\psi_{l}(z), l=1, 2, \cdots, n_2\}$.
Then construct the transfer function
$$
H(z)=\sum\limits_{k=1}^{n_1} \alpha_k z^{-(k-1)} + \sum\limits_{l=1}^{n_2} \beta_l \psi_{l}(z)
=[1, z^{-1}, \cdots, z^{-n_1+1}, \psi_{1}(z), \cdots, \psi_{n_2}(z)] \theta.
$$
Step 2:
Uniformly sample $z_1,z_2,\cdots,z_N$ on the upper unit circle,
and then randomly select $m$ samples to form the measurement model $H_{\Omega}=[\Phi \quad \Psi]_{\Omega} \theta$.

Step 3:
Use $\ell_1$ minimization (\ref{Eq5.1}) to solve the coefficient vector $\theta$.

The software $\ell_1$-magic is employed to recover the sparse coefficient $\theta$ from the two-ortho model.

The numbers of FIR basis and TM basis functions used are both 50,
with sparsity 3 and 2 under respective basis,
the number of measurements is $30=6(s_1+s_2)$ ($N$=4000),
and the experiment is repeated 100 times.
The relative reconstruction error of the sparse coefficients vector $\theta$ of the transfer function,
measured by
$
\frac{||\hat{\theta}-\theta||_2}{||\theta||_2},
$
is shown in Table \ref{tab:1},
with the maximal, minimal, average error of 100 trials in detail.

\begin{table}[H]
\centering
 \caption{\label{tab:1}The relative reconstruction error of sparse coefficients of the transfer function}
 \begin{tabular}{ccccc}
  \toprule
  \multirow{2}{*}{Model} & \multicolumn{3}{c} {Reconstruction error} & \multirow{2}{*}{Recover} \\
  \cmidrule(l){2-4}
  &  Max & Min & Average  &rate \\
  \midrule
 two-ortho   & 0.7004    & 6.3729e-006      & 0.0203      &   91 \%  \\
   \bottomrule
 \end{tabular}
\end{table}

The recover rate,
which is the percentage of reconstruction error below the given threshold out of the 100 trials,
is also given in Table \ref{tab:1}.
The threshold is set as 0.0005,
and the rate 91\%
shows that the optimization model (\ref{Eq5.1}) can accurately reconstruct the sparse coefficients
of the transfer function with high probability.

\subsection{Reconstruction of stable SISO system with poles known}
{\hspace{5mm}
This subsection gives examples to illustrate that
TM model \cite{XiChZh:14} and FIR model are special cases of the two-ortho model.
It shows that the redundant basis can handle the special case with only one basis.

\noindent
{\bf Example 1: TM model}

Consider an underlying SISO system with transfer function
\begin{equation}\label{Eq5.3}
H(z)  =   \frac{1.5 z - 0.871}{z^2 - 0.876 z + 0.00866}   =   \frac{1}{z-0.01}  +  \frac{1}{2z-\sqrt{3}}.
\end{equation}
We follow the aforementioned Steps 2 and 3
and repeat the experiment 100 times.
The numbers of FIR basis and TM basis functions used are both 100,
with sparsity 0 and 2 under respective bases.
The number of measurements is $28=14(s_1+s_2)$ ($N$=1000).
The first two poles in the TM basis are the true poles
and the others are all set to a constant, say 0.
The reconstruction error is measured by the $H_2$ norm of the difference between the original and reconstructed transfer functions.

The reconstruction performances,
including reconstruction error
and recover rate (for the threshold 0.0005),
are shown in Table \ref{tab:2}
and compared with the TM model (i.e. no FIR basis terms in two-ortho model) discussed in \cite{XiChZh:14},
whose sparsity and number of measurements are the same as those of two-ortho model.

\begin{table}[H]
\begin{center}
 \caption{\label{tab:2}Comparison between two-ortho and TM models}
 \begin{tabular}{ccccc}
  \toprule
  \multirow{2}{*}{Model} & \multicolumn{3}{c} {Reconstruction error} & \multirow{2}{*}{Recover} \\
  \cmidrule(l){2-4}
  &  Max & Min & Average  & rate\\
  \midrule
 two-ortho    &0.9274   &3.8716e-006    &0.0945      & 88\%  \\
   TM           &0.9415  &1.7261e-006   &0.0735      & 89\%   \\
   \bottomrule
 \end{tabular}
  \end{center}
\end{table}
As seen from Table \ref{tab:2},
the two-ortho model can accurately reconstruct the original transfer function with a high probability,
which is only slightly lower than that of the TM model for the concatenation of FIR bases.
This shows that TM model is indeed a special case of two-ortho model.

\noindent
{\bf Example 2: FIR model}

Next, we will give another example to show that
the sparse FIR model (i.e. no TM basis terms in two-ortho model) is also a special case of two-ortho model.

Consider the transfer function
$$H(z)=1/z^3+1/z^5+3/z^8.$$
The numbers of FIR basis and TM basis functions used are both 100,
with sparsity 3 and 0 under respective bases.
To avoid the TM basis degenerating to FIR basis,
the first three poles in TM basis are randomly chosen in the interval (0, 1)
and the others are all zero.
The number of measurements is $30=10(s_1+s_2)$ ($N$=1000).
The reconstruction performance is shown in Table \ref{tab:3},
including the $H_2$-norm of reconstruction error
and comparison to the sparse FIR model with the same sparsity and measurements as those of the two-ortho model.
The threshold used is 0.0005.

\begin{table}[H]
\centering
 \caption{\label{tab:3}Comparison between two-ortho and FIR models}
 \begin{tabular}{ccccc}
  \toprule
  \multirow{2}{*}{Model} & \multicolumn{3}{c} {Reconstruction error} & \multirow{2}{*}{Recover} \\
  \cmidrule(l){2-4}
  &  Max & Min & Average  & rate\\
  \midrule
 two-ortho    &0.5512e-003    &0.0063e-003    &0.0521e-003      & 99\%  \\
 FIR          &0.1838e-003    &0.0049e-003   &0.0286e-003   & 100\%   \\
   \bottomrule
 \end{tabular}
\end{table}
As seen from Table~\ref{tab:3},
two-ortho model and sparse FIR model can both accurately recover the transfer function with very high probability.
This shows that sparse FIR model is also a special case of two-ortho model.

\subsection{Reconstruction of stable SISO system with poles not exactly known}
{\hspace{5mm}
The above subsections are based on the assumption that
all the poles given in the TM basis are the true poles of the transfer function,
which is often unrealistic.
In this subsection, we consider the situations when the poles are partially known or fully unknown,
which is difficult to handle with either TM model or FIR model alone.

\noindent
{\bf Example 3: Reconstruction with part of the poles available}

Continue to consider the underlying SISO system (\ref{Eq5.3}).
The numbers of FIR basis and TM basis functions used are both 100,
with sparsity 3 and 2 under respective bases.
The number of measurements is $40=8(s_1+s_2)$ ($N$=1000).
Here we assume one of the poles $\sqrt{3}/2$ is known.
We repeated the experiment for 100 times
and compared with the sparse FIR model with 500 FIR basis functions in Table~\ref{tab:4}.

\begin{table}[H]  
\begin{center}
 \caption{\label{tab:4}Comparison between two-ortho model and FIR models}
 \begin{tabular}{cccccccc}
  \toprule
 \multirow{2}{*}{Model} & \multirow{2}{*}{Sparsity} &   \multirow{2}{*}{Measurement } & \multicolumn{3}{c}{ Reconstruction error} & \multirow{2}{*}{Recon. } & \multirow{2}{*}{Recover} \\
  \cmidrule(l){4-6}
  & &$\sharp$ & Max & Min & Average & Order & rate \\
  \midrule
 two-ortho   &  5     &  30  &0.9203    &0.0001    &0.0453    & 3  &94\% \\
 FIR         &  30    & 180     &0.2534    &0.0002    &0.0207  & 60 &45\%  \\
   \bottomrule
 \end{tabular}
 \end{center}
\end{table}


As shown in Table~\ref{tab:4},
the recover rate (for the threshold 0.0005) is
$94\%$ for two-ortho model.
This means that if one of the poles in the TM basis is the true pole ($\sqrt{3}/2$),
the optimization model (\ref{Eq5.1}) can reconstruct the original transfer function accurately with very high probability.

The reconstruction order in Table \ref{tab:4} is with respect to the minimum error from the 100 trials.
Although the maximal and average errors of sparse FIR model are less than those of two-ortho model,
 it needs more than 6 times measurements to get the reconstruction even with much higher order.
Hence, when the poles are partially known,
the two-ortho model has much better performance than sparse FIR model,
in terms of recover rate, number of measurements and reconstruction order.

\noindent
{\bf Example 4:  Reconstruction with all the poles unknown}

Consider the underlying SISO system (\ref{Eq5.3}).
Here the poles are all unknown,
we set different poles with different multiplicity to exemplify the performance of the two-ortho model.
The numbers of FIR basis and TM basis functions used are both 100,
with sparsity 3 and 6 under respective bases.
The number of measurements is $54=6(s_1+s_2)$ ($N$=1000).
We repeated the experiment for 100 times.
The reconstruction error of two-ortho model, with respect to different poles (0.9, 0.85, 0.8)
and multiplicity (2 to 6) in TM basis,
is shown in Fig. \ref{fig1}.

\begin{figure}[H]
\begin{center}
\includegraphics[height=6cm, width=\textwidth]{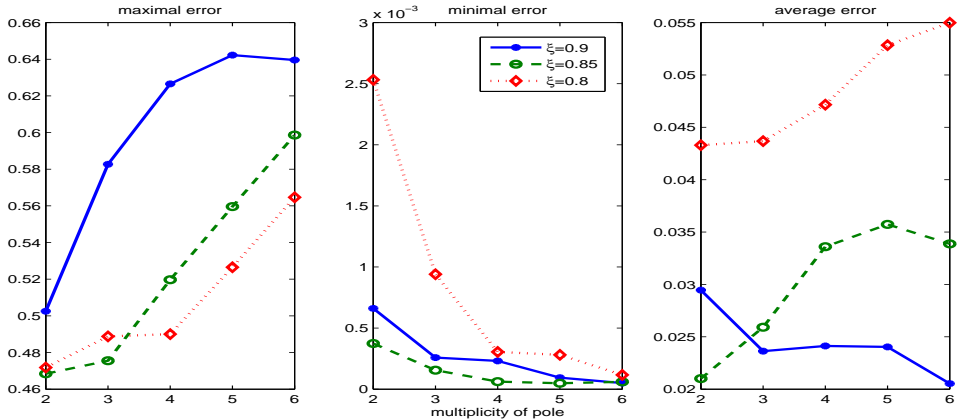}
\caption{Reconstruction error of two-ortho model.}
\label{fig1}
\end{center}
\end{figure}

As seen from Fig. \ref{fig1},
the maximal error increases with the multiplicity,
while the minimal error decreases with the multiplicity.
The average reconstruction error of the two-ortho model decreases as the multiplicity increases
when pole is 0.9,
while for pole 0.85 and 0.8, the average error increases.
This shows that higher multiplicity is necessary when the pole in TM basis is far away from the true pole.
The successful recover rate (for the threshold 0.0005)
is given in Table \ref{tab:5}.

\begin{table}[H]
\begin{center}
 \caption{\label{tab:5}Recover rate with respect to different poles and multiplicity}
 \begin{tabular}{cccccc}
  \toprule
 \multirow{2}{*}{Pole} & \multicolumn{5}{c} {Multiplicity}   \\
  \cmidrule(l){2-6}
& 2 &3    &4     &5 &6\\
  \midrule
 $\xi$ = 0.9  &0   &0.06  &0.33     &0.38 &0.31\\
 $\xi$ = 0.85 & 0.08 &0.66   &0.59     &0.43 &0.27\\
 $\xi$ = 0.8  & 0    &0   &0.06    &0.32 &0.39\\
   \bottomrule
 \end{tabular}
 \end{center}
\end{table}

From Fig. 1 and Table \ref{tab:5},
if the pole in TM basis is not the true pole of the transfer function,
we can use two-ortho model with a nearby  multiple pole  to reconstruct.
The farther the distance between the pole in TM basis and true pole,
the higher multiplicity is necessary.
And based on the recover rate,
there is an optimal multiplicity for each pole in TM basis.
The optimal multiplicity for the pole 0.85, 0.9, 0.8 is 3, 5, 6, respectively, and the corresponding reconstruction order with respect to the minimal error is 8, 5, 8, respectively.
While the order of reconstruction by sparse FIR model (see previous section) with the minimal error is 60,
and the recover rate is 45\%.
Hence, two-ortho model has better reconstruction performance
(reconstruction order, number of measurements and recover rate)
than sparse FIR model for the nearest pole 0.85,
and the farther pole will lead to lower recover rate.

\noindent
{\bf Example 5: Reconstruction with multiple poles}

We consider the situation when the transfer function has multiple poles,
which is difficult to handle with FIR model.
The transfer function of interest is
$$H(z)=\frac{1}{(z-0.1)^8}+\frac{(2-\sqrt{3}z)^4}{(2z-\sqrt{3})^5}.$$
The numbers of FIR basis and TM basis functions used are 100 and 50, respectively,
with sparsity of 3 and 2 under respective bases.
The number of measurements is $50=10(s_1+s_2)$ ($N$=1000).

{\bf 1) case 1: the dominant pole $\sqrt{3}/2$ is known.}

We set the first 5 poles in the TM basis as the true dominant poles,
and the others are all set to zero.
We repeat the experiment for 100 times.
The histogram of reconstruction error of two-ortho model,
compared with the FIR model with 500 FIR basis functions, is shown in Fig. 2.
As shown in Fig. \ref{fig2}, the errors from two-ortho model are mostly within (0, 0.1),
while the errors from FIR model scatter over (0, 1.2).
The parameters for further comparison are given in Table \ref{tab:6}.

\begin{figure}[H]
\begin{center}
\includegraphics[height=8cm, width=\textwidth]{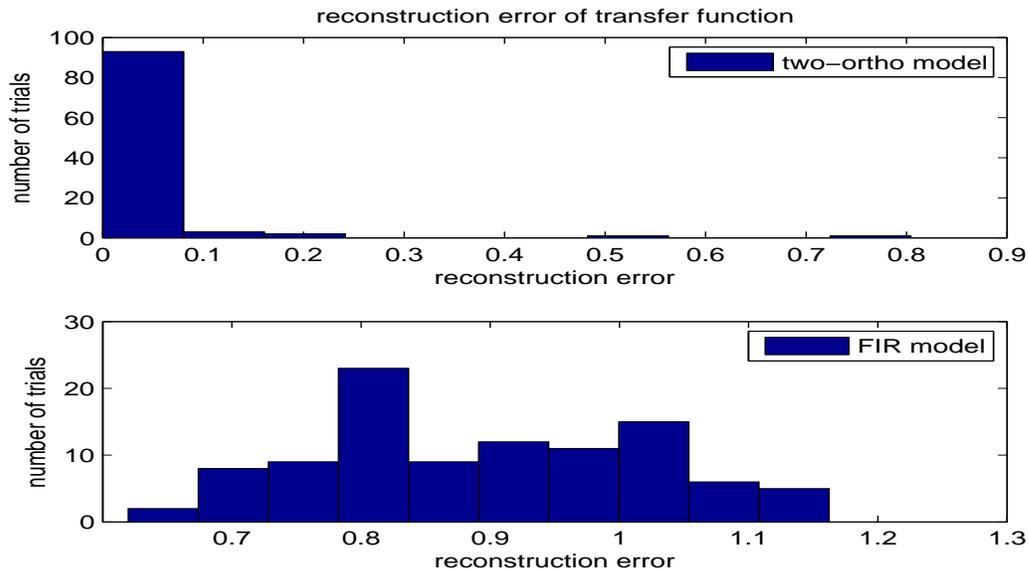}
\caption{Histogram of reconstruction error of two-ortho model and sparse FIR model.}
\label{fig2}
\end{center}
\end{figure}

\begin{table}[H]
\begin{center}
 \caption{\label{tab:6}Comparison between two-ortho and FIR models}
 \begin{tabular}{ccccccc}
  \toprule
 \multirow{2}{*}{Model} & \multirow{2}{*}{Sparsity} & \multirow{2}{*}{Measurements $\sharp$} & \multicolumn{3}{c}{Reconstruction error} & \multirow{2}{*}{Recon. Order} \\
  \cmidrule(l){4-6}
& &  &  Max & Min & Average & \\
  \midrule
two-ortho   &  5    & 50    & 0.8044    &0.0002    &0.0245     & 20  \\
FIR    &  30    & 120  &  1.1622   &0.6199    &0.8952     & 499 \\
   \bottomrule
 \end{tabular}
 \end{center}
\end{table}

The percentages of reconstruction error below 0.001 are $53\%$ and $0\%$ for two-ortho model and FIR model, respectively.
Hence two-ortho model is far better than FIR model for system with multiple poles.

{\bf 2) case 2: the poles of the original transfer function are unknown beforehand.}

The numbers of FIR and TM basis used are both 100, with the sparsity 3 and 2,
respectively and the number of measurements is 60.
The pole $\xi_1$ in TM basis is set as 0.9 with multiplicity 7 (higher than the true multiplicity),
and the others are all set to 0.
The histogram of reconstruction error with comparison to FIR model is shown in Fig. 3,
and the reconstruction performance is given in Table 7.
We also consider the situations when $\xi_1$ is set to 0.85 and 0.8, respectively,
the results are also given in Fig. \ref{fig3} and Table \ref{tab:7}.

\begin{figure}[H]
\begin{center}
\includegraphics[height=10cm, width=10cm]{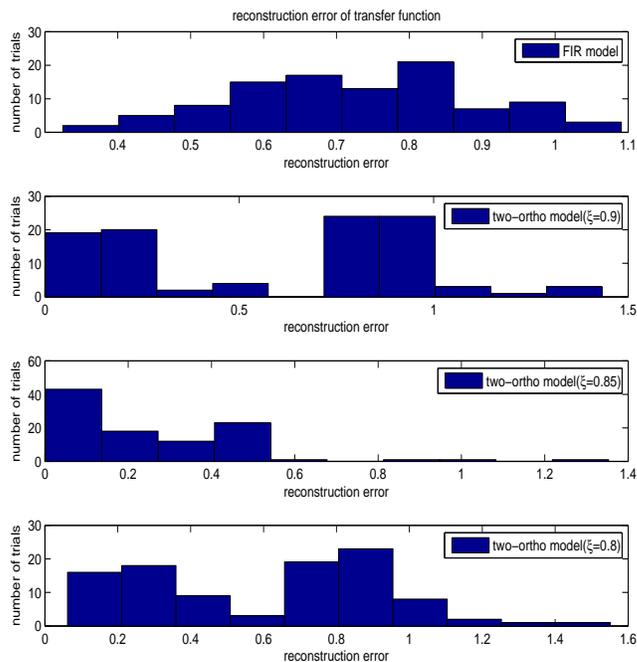}
\caption{Histogram of reconstruction error of two-ortho model and sparse FIR model (with multiple poles).}
\label{fig3}
\end{center}
\end{figure}

\begin{table}[H]
\begin{center}
 \caption{\label{tab:7}Recover rate of two-ortho and FIR models for threshold 0.005 with respect to different poles}
 \begin{tabular}{ccccccc}
  \toprule
  \multirow{2}{*}{Model}          &\multicolumn{3}{c}{Reconstruction error}  & \multirow{2}{*}{Recover }   & \multirow{2}{*}{Recon.}           \\
  \cmidrule(l){2-4}
& Max & Min & Average & rate & order\\
  \midrule
FIR     &  1.0906    &0.3247   & 0.72505    &0\%  & 499        \\
$\xi$ = 0.9    &1.4336    &0.0019   & 0.5756     &4\%    &  26        \\
$\xi$ = 0.85   &1.3534   & 0.0017    &0.2350    &13\%      &19    \\
$\xi$ = 0.8   & 1.5505   & 0.0621   & 0.6077    &0\%      & 106     \\
   \bottomrule
 \end{tabular}
 \end{center}
\end{table}

It is obvious from Fig. \ref{fig3} and Table \ref{tab:7}, that
two-ortho model and sparse FIR model both cannot handle well the system
without prior information about the multiple poles.
Even so,
the average reconstruction error and the reconstruction order of two-ortho model is less than those of sparse FIR model.
The closer the poles in the TM basis are to the true pole,
the less the average reconstruction error is,
and the order is much lower with a higher probability.

In application, the true transfer function is unknown,
so the reconstruction error cannot be calculated.
To choose a reconstruction with better performance,
the clustering criterion is based.
Here we cluster the reconstruction coefficients
from the 100 experiments when $\xi=0.85$.
The number of clusters is 10, and the histogram of clusters is shown in Fig. \ref{fig4}.

\begin{figure}[H]
\begin{center}
\includegraphics[height=5cm, width=10cm]{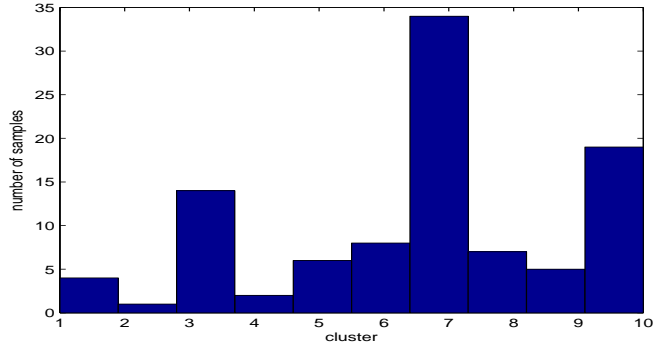}
\caption{Clusters of reconstruction coefficients.}
\label{fig4}
\end{center}
\end{figure}

From Fig. \ref{fig4}, we can see that the cluster 7 has most of the samples,
and the reconstruction order of these sample is almost 20.
Therefore, we can choose the reconstruction coefficient vector with the minimum error from this cluster.

\section{Conclusion}
{\hspace{5mm}
We have established the uniqueness of sparse representation
for a rational SISO system under FIR and TM bases.
With the uniqueness property and combining the principle of compressed sensing,
we have proposed an identification method using such two bases.
The identification method is linear in the two orthonormal bases,
and we have presented the lower bound of the number of measurement
which guarantees the efficient reconstruction
of sparse representation coefficients in pairs of FIR and TM bases by the $\ell_1$ minimization
from a limited number of observations on the upper unit circle with high probability.
It is shown that the proposed method can reconstruct the system
with much fewer measurements, much less error and higher probability
than those of FIR model, even without the prior information of poles.
Since the two-ortho model needs much fewer significant coefficients than those of FIR model,
it leads to a much lower order reconstruction than FIR model.



\end{document}